\documentclass[12pt]{article}
\usepackage{amsmath,amssymb,amsthm,bm}
\usepackage{graphicx}
\usepackage{enumerate}
\usepackage{natbib}
\usepackage{url}
\usepackage{booktabs,multirow,array}
\usepackage{caption}
\usepackage{algorithm}
\usepackage{algpseudocode}
\usepackage{lineno}
\usepackage[hidelinks]{hyperref}

\newcommand{\blind}{0}

\newtheorem{theorem}{Theorem}
\newtheorem{lemma}{Lemma}
\newtheorem{proposition}{Proposition}

\newcommand{\E}{\mathrm{E}}
\newcommand{\Ecc}{\mathrm{E}_{\mathrm{cc}}}
\newcommand{\Pcc}{P_{\mathrm{cc}}}
\newcommand{\Var}{\mathrm{Var}}
\newcommand{\bX}{\mathbf{X}}
\newcommand{\bx}{\mathbf{x}}
\newcommand{\bh}{\mathbf{h}}
\newcommand{\bt}{\mathbf{t}}
\newcommand{\bmu}{\bm{\mu}}
\newcommand{\bDelta}{\bm{\Delta}}
\newcommand{\bSigma}{\bm{\Sigma}}
\newcommand{\bOmega}{\bm{\Omega}}
\newcommand{\bLambda}{\bm{\Lambda}}
\newcommand{\btheta}{\bm{\theta}}
\newcommand{\bpsi}{\bm{\psi}}
\newcommand{\RR}{\mathrm{RR}}
\newcommand{\indep}{\perp\!\!\!\perp}
\newcommand{\odds}{\mathrm{odds}}
\newcommand{\post}{\varphi}

\begin{document}

\def\spacingset#1{\renewcommand{\baselinestretch}%
{#1}\small\normalsize} \spacingset{1}

%%%%%%%%%%%%%%%%%%%%%%%%%%%%%%%%%%%%%%%%%%%%%%%%%%%%%%%%%%%%%%%%%%%%%%%%%%%%%%

\if0\blind
{
  \title{\bf Three-Stage Weighting for Causal and Mediation Analysis in Case-Control Studies with Unknown Outcome Prevalence}
  \author{Tarikul Islam\thanks{
    The authors received no specific funding for this work.}\hspace{.2cm}\\
    Institute of Applied Statistics and Data Science (IASDS),\\ University of Dhaka, Dhaka-1000, Bangladesh\\
    and \\
    Mahbub A.H.M. Latif\thanks{Corresponding author. E-mail: \texttt{mlatif@isrt.ac.bd}} \\
    Institute of Applied Statistics and Data Science (IASDS),\\ University of Dhaka, Dhaka-1000, Bangladesh}
  \date{}
  \maketitle
} \fi

\if1\blind
{
  \bigskip
  \bigskip
  \bigskip
  \begin{center}
    {\LARGE\bf Three-Stage Weighting for Causal and Mediation Analysis in Case-Control Studies with Unknown Outcome Prevalence}
\end{center}
  \medskip
} \fi

\bigskip
\begin{abstract}
Population causal effects and their decomposition into direct and indirect effects cannot be estimated from case-control data without the outcome prevalence in the source population, which is rarely known. Because case-control sampling alters the population only through the outcome, the population is a mixture of the case and control distributions whose mixing weight is the prevalence. We show that this weight is identified by the population mean of any function of a single covariate that differs between cases and controls, and characterize what a reference distribution of the covariates identifies. This yields three prevalence estimators: moment calibration by the generalized method of moments, and two density-ratio fixed-point estimators using an individual-level reference sample or a synthetic sample drawn from the marginal distributions. Design weights based on the estimated prevalence reconstruct the population, and inverse probability weighting gives the total effect, the pure direct and indirect effects, and the mediated interaction. For the calibration estimator, we derive the joint asymptotic distribution of all stages and an overidentification test of the external information; we also quantify the sensitivity of the effects to the prevalence and extend the approach to frequency-matched designs. Simulations and an NHANES application illustrate the methods.
\end{abstract}

\noindent%
\vfill

\newpage
\spacingset{1.45} % DON'T change the spacing!
%\linenumbers

\section{Introduction}\label{sec:intro}

The case-control design remains one of the most efficient ways to study the determinants of an uncommon outcome, because sampling on the outcome concentrates the data collection effort where the information lies \citep{breslow1996statistics, rothman2008modern}. The price of this efficiency is that the sample no longer reflects the source population. The proportion of cases is fixed by the investigator, so the outcome distribution, and every quantity that depends on it, is distorted. A classical result of \citet{cornfield1951method} and \citet{prentice1979logistic} shows that conditional exposure odds ratios survive this distortion, which explains the dominance of logistic regression in case-control analysis. Contemporary epidemiological questions, however, are increasingly phrased in terms of marginal risks, risk ratios and causal mechanisms, and these are population quantities that depend on the outcome prevalence \citep{knol2008case, labrecque2021case}.

Two distinct sources of bias must therefore be handled at the same time. The first is the outcome-dependent sampling itself, and the second is confounding of the exposure, the mediator, and the outcome by common causes, which is present whenever exposure is not randomized. When the prevalence is known, the first problem has an elegant solution: weighting cases and controls by the ratio of their population and sample proportions reconstructs the source population, and standard causal estimators can then be applied within the reconstructed pseudo-population to handle the second problem \citep{van2008estimation, rose2008simple, rose2014double, maansson2007estimation, persson2017estimating}. \citet{l2022identification} formalized the identification conditions for this strategy, and a recent review documents a growing but still limited uptake of causal methods in case-control research \citep{mesidor2026use}.

For mediation analysis, the situation is more delicate, because the mediator model must also be estimated from the outcome-dependent sample. \citet{vanderweele2010odds} proposed fitting the mediator model among the controls, which is justified when the outcome is rare. \citet{satten2022efficient} developed a unified likelihood that also uses the cases in the mediator model, gaining efficiency for the indirect effect on the odds ratio scale while retaining a rare-disease approximation for the controls. \citet{doretti2024mediation} derived the exact distortion that case-control sampling induces in a logistic model for a binary mediator, which yields exact identification of natural effects without a rare-disease assumption but requires the population prevalence, possibly within strata, to be known from external sources. \citet{caubet2024empirical} empirically compared approximate and exact regression-based approaches. The available methods therefore either rely on a rare-outcome approximation or treat the prevalence as known.

In practice, the prevalence of the outcome in the specific source population, for example, the catchment population of the participating hospitals during the study period, is frequently unavailable. In contrast, summaries of basic covariates of that population, such as its age and sex composition or its mean body mass index, are routinely published by statistical agencies and population surveys.  Without the prevalence, causal risk ratios are only partially identified \citep{gabriel2022causal}, and a sensitivity analysis over a range of plausible values is the usual fallback. The idea of combining an outcome-dependent sample with aggregate population information has a long history in econometrics, where choice-based samples are combined with known population shares or with supplementary random samples \citep{manski1977estimation, hsieh1985estimation, imbens1994combining}, in biostatistics, where summary-level external information is used to improve the efficiency of regression estimates \citep{chatterjee2016constrained, qin2015using}, and in machine learning, where estimating new class priors from unlabelled covariates is known as the label-shift problem \citep{saerens2002adjusting, lipton2018detecting}. To our knowledge, these ideas have not been used to identify and estimate causal mediation effects from case-control data with unknown prevalence.

This article builds a complete inferential framework on one observation.  Because case-control sampling selects individuals at random within outcome strata, the likelihood ratio between the source population and the sample is a function of the outcome only, and that function is known up to the prevalence. The source population is therefore a mixture of the case and control distributions, both of which the sample represents faithfully, with the prevalence as mixing weight. We use this observation in three ways, depending on the kind of external information available. When population means of covariate functions are known, the prevalence is a linear function of those means and the case and control means, and we estimate it by calibrating the generalized method of moments \citep{hansen1982large}; we need neither the joint distribution of the covariates nor an outcome model. When an individual-level reference sample of the covariates is available, the prevalence is the maximizer of a mixture likelihood and can be computed by a density-ratio fixed-point iteration. When only the univariate marginal distributions are available, the same iteration can be applied to a synthetic sample drawn from their product, and we characterize exactly when this yields the prevalence. The estimated prevalence defines design weights that reconstruct the source population, and inverse probability weights for the exposure and the mediator then give the nested counterfactual means that define the total effect and its three-way decomposition into pure direct, pure indirect, and mediated interaction components \citep{vanderweele2013three}. For the calibration estimator, we derive the joint asymptotic distribution of all three stages under the stratified sampling design, including the uncertainty of the external moments, and an overidentification test that diagnoses external information that does not refer to the source population. We also characterize analytically when the prevalence matters for ratio effects and extend the approach to frequency-matched designs.

The remainder of the article is organized as follows. Section~\ref{sec:ident} introduces the setting, the causal estimands and the identification results. Section~\ref{sec:est} presents the three estimators of the prevalence, the weighting estimators of the effects, their asymptotic properties and the extension to matched designs. Section~\ref{sec:sim} reports a simulation study, Section~\ref{sec:app} presents an application to data from the National Health and Nutrition Examination Survey, and Section~\ref{sec:disc} concludes with a discussion. Proofs, additional simulation results, and code are provided in the Supplementary Materials.

\section{Setting, Estimands and Identification}\label{sec:ident}

\subsection{Case-control sampling and design weights}\label{sec:setting}

Let $O=(\bX, A, M, Y)$ denote the data of an individual in the source population, where $\bX=(X_1,\ldots,X_p)^\top$ is a vector of pre-exposure covariates, $A\in\{0,1\}$ an exposure, $M\in\{0,1\}$ a mediator and $Y\in\{0,1\}$ the outcome. We write $P$ and $\E$ for the distribution and expectation in the source population and $\pi=P(Y=1)\in(0,1)$ for the outcome prevalence. The case-control sample consists of $n_1$ cases drawn at random from the subpopulation with $Y=1$ and $n_0$ controls drawn at random from the subpopulation with $Y=0$, independently, with $n=n_1+n_0$ and case fraction $\lambda=n_1/n$ fixed by design. It is convenient to describe the sample through the distribution of a randomly chosen sampled individual, $\Pcc=\lambda P(\cdot\mid Y=1)+(1-\lambda)P(\cdot\mid Y=0)$, with expectation $\Ecc$. The defining feature of the design is the following assumption.

\noindent\textit{Assumption A1 (outcome-dependent sampling).} Within each outcome stratum, the sample is representative of the source population, that is, $\Pcc(O\in\cdot\mid Y=y)=P(O\in\cdot\mid Y=y)$ for $y=0,1$.

Assumption A1 is equivalent to requiring that the probability of selection depends on $(\bX,A,M,Y)$ only through $Y$. It fails if, for example, controls are recruited from a population that differs from the one that generated the cases. Define the design weight
\begin{equation}
\omega(y;\pi)=\frac{\pi}{\lambda}\,I(y=1)+\frac{1-\pi}{1-\lambda}\,I(y=0), \label{eq:omega}
\end{equation}
the ratio of the population and sample proportions of the outcome category $y$. The following lemma shows that it is the exact likelihood ratio between the two distributions, so that reweighting by $\omega$ reconstructs not only the outcome distribution but the entire joint distribution of $O$.

\begin{lemma}\label{lem:rn}
Under Assumption A1, $P$ is absolutely continuous with respect to $\Pcc$ with density $dP/d\Pcc(o)=\omega(y;\pi)$. Consequently, $\E\{g(O)\}=\Ecc\{\omega(Y;\pi)g(O)\}$ for every integrable function $g$, and $\E\{g(O)\mid Z\}=\Ecc\{\omega(Y;\pi)g(O)\mid Z\}/\Ecc\{\omega(Y;\pi)\mid Z\}$ for every sub-vector $Z$ of $O$.
\end{lemma}

Lemma~\ref{lem:rn} has three consequences that organize the rest of the paper. First, once $\pi$ is known, every population quantity, including propensity scores, mediator distributions and counterfactual means, is a functional of the design-weighted sample distribution. Second, the only unknown ingredient of the reweighting is the scalar $\pi$; the joint distribution of the covariates in the population need not be modelled, because it is carried, stratum by stratum, by the sample itself. Third, $\Ecc\{\omega(Y;\tilde\pi)Y\}=\tilde\pi$ for any working value $\tilde\pi$, so reweighting reproduces whatever prevalence is assumed, and the case-control sample alone contains no information about the prevalence. Applied with $g=Y$ and $Z=\bX$, the lemma also yields the familiar relation
\begin{equation}
\odds\{P(Y=1\mid\bX)\}=\alpha(\pi)\,\odds\{\Pcc(Y=1\mid\bX)\},\qquad \alpha(\pi)=\frac{\pi/(1-\pi)}{\lambda/(1-\lambda)}, \label{eq:alpha}
\end{equation}
between the population and case-control risk functions \citep{prentice1979logistic}.

\subsection{Causal estimands}\label{sec:estimands}

We use potential outcome notation throughout. Let $M(a)$ denote the mediator that would be observed if the exposure were set to $a$, and let $Y(a,m)$ denote the outcome that would be observed if the exposure were set to $a$ and the mediator to $m$. The nested potential outcome $Y\{a,M(a^\ast)\}$ is the outcome under exposure $a$ when the mediator takes the value it would have taken under exposure $a^\ast$, and we write
\[
\nu(a,a^\ast)=\E[\,Y\{a,M(a^\ast)\}\,],\qquad a,a^\ast\in\{0,1\},
\]
for the four corresponding population risks. By consistency, $\nu(a,a)=\E\{Y(a)\}$ is the risk under exposure $a$. The four effects considered in this paper are ratios of these risks. The total effect is $\RR^{\mathrm{TE}}=\nu(1,1)/\nu(0,0)$, the pure direct effect is $\RR^{\mathrm{PDE}}=\nu(1,0)/\nu(0,0)$, the pure indirect effect is $\RR^{\mathrm{PIE}}=\nu(0,1)/\nu(0,0)$, and the mediated interaction is
\[
\RR^{\mathrm{InE}}=\frac{\nu(1,1)\,\nu(0,0)}{\nu(1,0)\,\nu(0,1)}.
\]
These definitions imply the multiplicative three-way decomposition $\RR^{\mathrm{TE}}=\RR^{\mathrm{PDE}}\times\RR^{\mathrm{PIE}}\times\RR^{\mathrm{InE}}$, the ratio analogue of the decomposition of \citet{vanderweele2013three}. The pure direct effect compares exposure levels with the mediator held at its distribution under no exposure, the pure indirect effect changes only the mediator distribution with the exposure held at zero, and the mediated interaction equals one when the exposure does not modify the effect of the mediator on the multiplicative scale. All four estimands are defined in the source population and not in the case-control sample.

\subsection{Identification of the prevalence}\label{sec:identpi}

The source population can be written as $P=\pi P(\cdot\mid Y=1)+(1-\pi)P(\cdot\mid Y=0)$, and by Assumption A1 both components are represented in the case-control sample. Identifying $\pi$ therefore amounts to identifying the mixing weight of a two-component mixture whose components are known, and this requires some information on the population distribution of the covariates. We consider two kinds of such information.

The first consists of population means of covariate functions. Suppose that for a vector $\bh(\bX)=\{h_1(\bX),\ldots,h_q(\bX)\}^\top$ of known functions the population mean $\bt=\E\{\bh(\bX)\}$ is available from a source external to the case-control study. In the applications we have in mind, each $h_k$ depends on a single covariate, for instance the indicator of female sex, age, or the square of body mass index, so that $\bt$ consists of proportions, means and second moments of univariate marginal distributions. Let $\bmu_y=\E\{\bh(\bX)\mid Y=y\}$ and $\bDelta=\bmu_1-\bmu_0$. By Assumption A1, $\bmu_y=\Ecc\{\bh(\bX)\mid Y=y\}$, so both case and control means are estimable from the sample. The following assumptions complete the identification argument.

\noindent\textit{Assumption A2 (transportable external information).} The external moments refer to the source population of the case-control study, $\bt=\E\{\bh(\bX)\}$, and $\E\|\bh(\bX)\|^2<\infty$.

\noindent\textit{Assumption A3 (relevance).} $\bDelta\neq\mathbf{0}$.

Assumption A2 is the substantive requirement and is discussed further in Sections~\ref{sec:est} and \ref{sec:disc}. Assumption A3 requires that at least one calibration function differs in mean between cases and controls; it plays the role of the relevance condition for an instrument and can be checked in the data.

\begin{theorem}\label{thm:pi}
Under Assumption A1, $\bt=\pi\bmu_1+(1-\pi)\bmu_0$, and the following statements hold.
\begin{enumerate}
\item[(i)] Under Assumptions A2 and A3, $\pi$ is identified, and for every positive definite $q\times q$ matrix $\mathbf{W}$,
\begin{equation}
\pi=\frac{\bDelta^\top\mathbf{W}(\bt-\bmu_0)}{\bDelta^\top\mathbf{W}\bDelta}. \label{eq:pi_ident}
\end{equation}
For a single calibration function, this reduces to $\pi=(t-\mu_0)/(\mu_1-\mu_0)$.
\item[(ii)] If $\bDelta=\mathbf{0}$, or if no external information is available, $\pi$ is not identified: for every $\pi'\in(0,1)$ there is a source population with prevalence $\pi'$ that satisfies Assumption A1 and generates the same distribution $\Pcc$.
\item[(iii)] When $q>1$, Assumptions A1 and A2 imply the testable restriction that $\bt-\bmu_0$ is proportional to $\bDelta$.
\end{enumerate}
\end{theorem}

Because each calibration function may depend on a single covariate, univariate marginal moments are sufficient, and correlation among the covariates is irrelevant for identification: the dependence structure is carried by the case and control samples. Part (iii) implies that, with more than one calibration function, the external information can be partially validated.

The second kind of information is a reference distribution $Q$ for the covariates, available either through an individual-level reference sample from the source population, in which case $Q$ is the population distribution $P_{\bX}$ of $\bX$, or through the univariate marginal distributions of $X_1,\ldots, X_p$, from which one can construct the product distribution $Q=P_{X_1}\otimes\cdots\otimes P_{X_p}$. Let $f_y$ denote the density of $\bX$ given $Y=y$ with respect to a dominating measure, and for a working value $\tilde\pi$ let
\[
\post(\bx;\tilde\pi)=\frac{\tilde\pi f_1(\bx)}{\tilde\pi f_1(\bx)+(1-\tilde\pi)f_0(\bx)}
\]
be the risk function of a population with prevalence $\tilde\pi$ and the same case and control distributions. By (\ref{eq:alpha}), $\post(\bx;\tilde\pi)$ is obtained from the case-control risk function by multiplying its odds by $\alpha(\tilde\pi)$, so it is a functional of $\Pcc$ for every $\tilde\pi$. Since $\pi=\E\{P(Y=1\mid\bX)\}$, the prevalence solves the fixed-point equation $\tilde\pi=\E_Q\{\post(\bX;\tilde\pi)\}$ when $Q=P_{\bX}$. The following result describes the solutions for a general $Q$.

\begin{proposition}\label{prop:fixed}
Let $Q$ be a distribution for $\bX$ that is absolutely continuous with respect to the case-control distribution of $\bX$, and define $\ell_Q(\tilde\pi)=\E_Q\log\{\tilde\pi f_1(\bX)+(1-\tilde\pi)f_0(\bX)\}$ for $\tilde\pi\in(0,1)$. Assume that $Q(f_1\neq f_0)>0$ and that $\ell_Q$ is finite on $(0,1)$.
\begin{enumerate}
\item[(i)] $\ell_Q$ is strictly concave with derivative $\ell_Q'(\tilde\pi)=[\E_Q\{\post(\bX;\tilde\pi)\}-\tilde\pi]/\{\tilde\pi(1-\tilde\pi)\}$. Hence the equation $\tilde\pi=\E_Q\{\post(\bX;\tilde\pi)\}$ has at most one solution in $(0,1)$; the one-sided limits of the derivative are $\ell_Q'(0+)=\E_Q(f_1/f_0)-1$ and $\ell_Q'(1-)=1-\E_Q(f_0/f_1)$, a solution exists if and only if $\ell_Q'(0+)>0$ and $\ell_Q'(1-)<0$, and it is then the maximizer $\pi_Q$ of $\ell_Q$.
\item[(ii)] If $Q=P_{\bX}$, the solution exists and equals $\pi$.
\item[(iii)] In general $\pi_Q$ is the mixing weight of the mixture $\tilde\pi f_1+(1-\tilde\pi)f_0$ closest to $Q$ in Kullback-Leibler divergence, and $\pi_Q=\pi$ whenever $Q=P_{\bX}$. For the product of the marginals, $\pi_Q=\pi$ when the covariates are independent in the population, but not in general otherwise, and when the maximizer of $\ell_Q$ lies on the boundary of $[0,1]$, no solution exists in $(0,1)$.
\end{enumerate}
\end{proposition}

The fixed-point iteration $\tilde\pi\leftarrow\E_Q\{\post(\bX;\tilde\pi)\}$ is the expectation-maximization algorithm for the mixing weight of a two-component mixture with known components \citep{saerens2002adjusting}, and Proposition~\ref{prop:fixed} shows that it converges to the prevalence when the reference distribution is the population distribution of the covariates. The proposition also shows the price of using the product of the marginals as reference: the procedure then targets the mixture closest to an independence approximation of the population, which coincides with the population only under independence.

\subsection{Identification of the causal and mediation estimands}\label{sec:identeff}

Let $e_a(\bx)=P(A=a\mid\bX=\bx)$ denote the population propensity score and $f(m\mid a,\bx)=P(M=m\mid A=a,\bX=\bx)$ the population mediator distribution. We make the standard assumptions for natural effects \citep{imai2010identification, vanderweele2014effect}. Assumption C1 (consistency) states that $M=M(a)$ when $A=a$ and $Y=Y(a,m)$ when $A=a$ and $M=m$. Assumption C2 (no unmeasured confounding) states that, for all $a$, $a^\ast$ and $m$, $Y(a,m)\indep A\mid\bX$, $Y(a,m)\indep M\mid(A,\bX)$, $M(a^\ast)\indep A\mid\bX$ and $Y(a,m)\indep M(a^\ast)\mid\bX$. The last of these conditions is a cross-world independence which, in particular, excludes mediator-outcome confounders affected by the exposure \citep{vanderweele2014effect}. Assumption C3 (positivity) requires $e_a(\bX)>0$ and $f(m\mid a,\bX)>0$ almost surely for all $a$ and $m$. Define the weight function
\begin{equation}
W_{a,a^\ast}(O)=\frac{I(A=a)}{e_a(\bX)}\,\frac{f(M\mid a^\ast,\bX)}{f(M\mid a,\bX)}, \label{eq:Wfun}
\end{equation}
whose mediator factor equals one when $a=a^\ast$.

\begin{theorem}\label{thm:eff}
Under Assumptions A1 and C1 to C3,
\begin{equation}
\nu(a,a^\ast)=\frac{\Ecc\{\omega(Y;\pi)\,W_{a,a^\ast}(O)\,Y\}}{\Ecc\{\omega(Y;\pi)\,W_{a,a^\ast}(O)\}}, \label{eq:nu_ident}
\end{equation}
where $e_a$ and $f$ are themselves design-weighted functionals of $\Pcc$, namely
\begin{gather*}
e_a(\bx)=\frac{\Ecc\{\omega(Y;\pi)I(A=a)\mid\bX=\bx\}}{\Ecc\{\omega(Y;\pi)\mid\bX=\bx\}},\\
f(m\mid a,\bx)=\frac{\Ecc\{\omega(Y;\pi)I(M=m)\mid A=a,\bX=\bx\}}{\Ecc\{\omega(Y;\pi)\mid A=a,\bX=\bx\}}.
\end{gather*}
Hence, if $\pi$ is identified, for example under the conditions of Theorem~\ref{thm:pi} or Proposition~\ref{prop:fixed}(ii), the four effects of Section~\ref{sec:estimands} are identified.
\end{theorem}

Theorem~\ref{thm:eff} makes precise the sense in which the design weights create a pseudo-population in which standard weighting methods are valid. The propensity score and the mediator model must be estimated with the design weights, because the unweighted case-control sample overrepresents the exposure and mediator patterns that are associated with the outcome. The denominator of (\ref{eq:nu_ident}) equals one in the population, but normalizing by its sample analogue leads to a more stable estimator.

\subsection{How much do the effects depend on the prevalence?}\label{sec:sens}

The risk ratio estimands depend on the prevalence, but the strength of this dependence varies. To make this precise, fix a nonnegative weight function $W(O)$, such as (\ref{eq:Wfun}) with its nuisance functions held fixed, and consider the functional,
\[
\nu_W(\tilde\pi)=\frac{\Ecc\{\omega(Y;\tilde\pi)W(O)Y\}}{\Ecc\{\omega(Y;\tilde\pi)W(O)\}}
\]
obtained when the design weights are computed with a working prevalence $\tilde\pi$.

\begin{proposition}\label{prop:sens}
Let $W(O)\ge0$ satisfy $\Ecc\{W(O)\mid Y=y\}>0$ for $y=0,1$. For every $\tilde\pi\in(0,1)$,
\[
\odds\{\nu_W(\tilde\pi)\}=\odds(\tilde\pi)\,\frac{\Ecc\{W(O)\mid Y=1\}}{\Ecc\{W(O)\mid Y=0\}}.
\]
Consequently, for two such weight functions $W$ and $W'$:
\begin{enumerate}
\item[(i)] the odds ratio $\odds\{\nu_W(\tilde\pi)\}/\odds\{\nu_{W'}(\tilde\pi)\}$ does not depend on $\tilde\pi$;
\item[(ii)] if $\tilde\pi$ replaces the true prevalence $\pi$, then $\nu_W(\tilde\pi)=\zeta\nu_W(\pi)/\{1-\nu_W(\pi)+\zeta\nu_W(\pi)\}$ with $\zeta=\odds(\tilde\pi)/\odds(\pi)$;
\item[(iii)] $\partial\log\{\nu_W(\tilde\pi)/\nu_{W'}(\tilde\pi)\}/\partial\tilde\pi=\{\nu_{W'}(\tilde\pi)-\nu_W(\tilde\pi)\}/\{\tilde\pi(1-\tilde\pi)\}$.
\end{enumerate}
\end{proposition}

Part (i) is the counterfactual counterpart of the classical invariance of the odds ratio under case-control sampling. Part (iii) shows that the sensitivity of a log risk ratio to the prevalence equals minus the corresponding risk difference divided by $\pi(1-\pi)$. A working prevalence that is wrong by the relative amount $(\tilde\pi-\pi)/\pi$ therefore shifts the log risk ratio by approximately $-\{\mathrm{RD}/(1-\pi)\}(\tilde\pi-\pi)/\pi$, where RD is the counterfactual risk difference. When the outcome is rare, the counterfactual risks and their difference are small, so even a sizeable relative error in the prevalence has little effect on risk ratio effects; when the outcome is common, an accurate prevalence matters. Part (ii) provides an exact, closed-form sensitivity analysis: given counterfactual risks computed under one value of the prevalence, those under any other value are obtained by rescaling their odds.

\section{Estimation and Inference}\label{sec:est}

\subsection{Overview}\label{sec:overview}

The identification results translate into a three-stage estimator, summarized in Figure~\ref{fig:scheme}. Stage 1 estimates the prevalence from the case-control sample and the external information; Stage 2 converts the estimated prevalence into design weights; and Stage 3 estimates the counterfactual risks by inverse probability weighting in the reconstructed population. Three estimators are available for Stage 1, corresponding to the three kinds of external information discussed in Section~\ref{sec:identpi}; Stages 2 and 3 are common to all of them.

\begin{figure}[t]
\centering
\includegraphics[width=\textwidth]{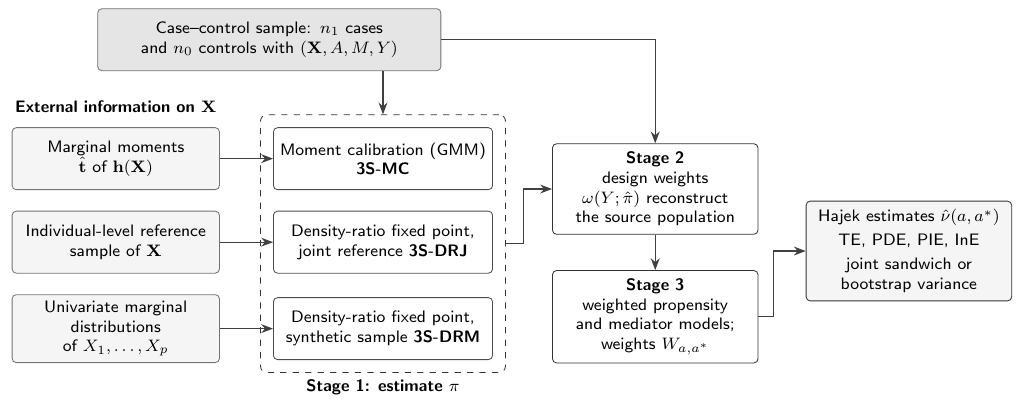}
\caption{Structure of the three-stage weighting framework. Stage 1 uses the case-control sample together with one of three kinds of external information on the covariates; Stages 2 and 3 operate on the case-control sample with the design weights implied by the estimated prevalence $\hat\pi$.}
\label{fig:scheme}
\end{figure}

\subsection{Stage 1: three estimators of the prevalence}\label{sec:stage1}

\paragraph{Moment calibration (3S-MC).} Let $\hat\bmu_1$ and $\hat\bmu_0$ denote the sample means of $\bh(\bX)$ among cases and controls, $\hat\bDelta=\hat\bmu_1-\hat\bmu_0$, and let $\hat\bt$ be the external estimate of $\bt$, with covariance matrix $\mathbf{V}_t$ when it comes from a sample survey and $\mathbf{V}_t=\mathbf{0}$ when it is known exactly, as for census totals. The sample analogue of Theorem~\ref{thm:pi} is the moment function $\hat{\mathbf g}(\tilde\pi)=\hat\bt-\hat\bmu_0-\tilde\pi\hat\bDelta$, which is linear in $\tilde\pi$, and minimizing $\hat{\mathbf g}(\tilde\pi)^\top\mathbf{W}\hat{\mathbf g}(\tilde\pi)$ gives the closed-form generalized method of moments estimator
\begin{equation}
\hat\pi_{\mathbf{W}}=\frac{\hat\bDelta^\top\mathbf{W}(\hat\bt-\hat\bmu_0)}{\hat\bDelta^\top\mathbf{W}\hat\bDelta}. \label{eq:pihat}
\end{equation}
Equivalently, $\hat\pi$ is the value for which the design-weighted sample reproduces the external moments, in the sense that $n^{-1}\sum_i\omega(Y_i;\hat\pi)\,\mathbf{c}^\top\{\bh(\bX_i)-\hat\bt\}=0$ with $\mathbf{c}=\mathbf{W}\hat\bDelta$; the estimator is therefore a one-parameter calibration estimator in the sense of \citet{deville1992calibration}, in which the calibration weights are constrained to the form dictated by the design. The efficient choice of $\mathbf W$ is the inverse of
\begin{equation}
\hat\bOmega(\tilde\pi)=\tilde\pi^2\frac{\hat\bSigma_1}{n_1}+(1-\tilde\pi)^2\frac{\hat\bSigma_0}{n_0}+\mathbf{V}_t, \label{eq:Omega}
\end{equation}
where $\hat\bSigma_y$ is the sample covariance matrix of $\bh(\bX)$ in outcome stratum $y$. Since $\hat\bOmega$ depends on $\tilde\pi$, we start from a diagonal weight matrix and iterate (\ref{eq:pihat}) with $\mathbf W=\hat\bOmega(\hat\pi)^{-1}$; three iterations sufficed in all our analyses. With $q>1$ calibration functions, the statistic $J=\hat{\mathbf g}(\hat\pi)^\top\hat\bOmega(\hat\pi)^{-1}\hat{\mathbf g}(\hat\pi)$ tests the restriction of Theorem~\ref{thm:pi}(iii). A large value of $J$ indicates that the external moments are incompatible with any population from which the cases and controls could have been drawn, for example because they refer to a different population or are measured differently; the test has no power against errors that shift all moments in the direction of $\bDelta$, which are indistinguishable from a change in $\pi$. When the outcome is very rare, or the covariates discriminate weakly, the unconstrained estimate (\ref{eq:pihat}) may fall outside the unit interval. We then use the constrained solution, which, because the objective is quadratic, is the projection of (\ref{eq:pihat}) onto $[\epsilon,1-\epsilon]$ with a small $\epsilon$ ($0.001$ in our analyses); we treat $\pi$ as fixed in Stages 2 and 3 in that case and regard the event as a signal of weak identification.

\paragraph{Density ratio with a joint reference sample (3S-DRJ).} Suppose that an individual-level sample $\bX^\ast_1,\ldots,\bX^\ast_{N}$ of covariates from the source population is available, for example from a population survey. Proposition~\ref{prop:fixed}(ii) suggests estimating $\pi$ as the solution of a sample version of $\tilde\pi=\E_Q\{\post(\bX;\tilde\pi)\}$ with $Q=P_{\bX}$. The risk function $\post(\cdot;\tilde\pi)$ is estimated by fitting a logistic regression $\hat S(\bx)$ of $Y$ on $\bX$ in the case-control sample and multiplying its odds by $\alpha(\tilde\pi)$ as in (\ref{eq:alpha}), which gives $\hat\post(\bx;\tilde\pi)=\alpha(\tilde\pi)\hat S(\bx)/\{1-\hat S(\bx)+\alpha(\tilde\pi)\hat S(\bx)\}$. The expectation under $Q$ is estimated by reweighting the case-control sample with the density ratio $r(\bx)=dQ/d\Pcc(\bx)$ between the reference and case-control covariate distributions. We estimate $r$ by a logistic classifier $\hat R(\bx)$ that discriminates the reference sample ($Z=1$) from the case-control covariates ($Z=0$) and set $\hat r(\bx)=\hat R(\bx)/\{1-\hat R(\bx)\}$. The estimator $\hat\pi_{\mathrm{DRJ}}$ is the solution in $(0,1)$ of
\begin{equation}
\tilde\pi=\frac{\sum_{i=1}^n\hat r(\bX_i)\,\hat\post(\bX_i;\tilde\pi)}{\sum_{i=1}^n\hat r(\bX_i)}, \label{eq:fixedpoint}
\end{equation}
which we compute by one-dimensional root finding. When the classifier is correctly specified, the ratio of the class sizes $N/n$ only shifts its intercept, which multiplies the estimated odds by a constant that cancels in the self-normalized weights of (\ref{eq:fixedpoint}); under misspecification, the logistic projection that the classifier estimates, and therefore the limit of the estimator, depends on $N/n$. The estimator is consistent when the risk model and the classifier are correctly specified, since the limit of (\ref{eq:fixedpoint}) is then the equation of Proposition~\ref{prop:fixed}(ii). When individual-level reference data are available, $\E_Q\{\post(\bX;\tilde\pi)\}$ could also be estimated by averaging $\hat\post$ over the reference sample; the density-ratio form is used here because it extends directly to the next estimator.

\paragraph{Density ratio with a synthetic marginal sample (3S-DRM).} When only the univariate marginal distributions of the covariates are known, a synthetic reference sample of size $N_{\mathrm{syn}}$ can be drawn from their product $Q=P_{X_1}\otimes\cdots\otimes P_{X_p}$, and (\ref{eq:fixedpoint}) can be solved with $\hat r$ estimated against that sample. When the risk model and the classifier are correctly specified, Proposition~\ref{prop:fixed}(iii) implies that $\hat\pi_{\mathrm{DRM}}$ converges to $\pi_Q$, which equals $\pi$ when the covariates are independent in the population but generally differs from $\pi$ otherwise, and the equation may then have no solution in $(0,1)$. With misspecified working models, the limit of $\hat\pi_{\mathrm{DRM}}$ is the solution of the corresponding working-model equation, which can differ from both $\pi$ and $\pi_Q$. When the classifier is correctly specified, the synthetic sample size affects only the Monte Carlo error and the finite-sample behaviour of the classifier; with a misspecified classifier, the limit also depends on the ratio $N_{\mathrm{syn}}/n$, as the simulations of Section~\ref{sec:sim} illustrate. The estimator is therefore a convenient approximation that uses the same external information as 3S-MC in distributional form, and its accuracy depends on how far the covariates are from independence.

The three estimators thus use increasingly weak forms of external information, and the assumptions they require differ accordingly. 3S-MC requires only marginal moments, no outcome model and no density-ratio model, and is valid whatever the dependence among the covariates. 3S-DRJ requires individual-level reference data and correct specification of the risk model and the classifier, and in exchange it uses the whole reference distribution. 3S-DRM requires the same working models as 3S-DRJ and, in addition, approximate independence of the covariates. In practice, we recommend 3S-MC, whose moments can also be computed from an individual-level reference sample when one is available, and we use the density-ratio estimators as complements and as diagnostics, with a large discrepancy between them and 3S-MC signalling misspecification of the working models or dependence among the covariates.

\subsection{Stages 2 and 3: design weights and weighting estimators of the effects}\label{sec:stage23}

In Stage 2, each sampled individual receives the design weight $\hat\omega_i=\omega(Y_i;\hat\pi)$, that is, $\hat\pi/\lambda$ for cases and $(1-\hat\pi)/(1-\lambda)$ for controls, where $\hat\pi$ is any of the three estimators of Section~\ref{sec:stage1} and $\lambda=n_1/n$ is known by design. The weighted sample has outcome prevalence $\hat\pi$ and, by Lemma~\ref{lem:rn}, mimics the source population in all other respects.

In Stage 3, the propensity score $e_1(\bx;\bm{\beta})=\mathrm{expit}(\tilde\bx^\top\bm{\beta})$, with $\tilde\bx=(1,\bx^\top)^\top$, is fitted by logistic regression with weights $\hat\omega_i$, and the mediator model $P(M=1\mid A=a,\bX=\bx;\bm{\gamma})=\mathrm{expit}\{(\tilde\bx^\top,a)\bm{\gamma}\}$ is fitted in the same way. Plugging the fitted models into (\ref{eq:Wfun}) gives weights $\hat W_{a,a^\ast, i}$, and the counterfactual risks are estimated by the normalized weighted means
\begin{equation}
\hat\nu(a,a^\ast)=\frac{\sum_{i=1}^n\hat\omega_i\hat W_{a,a^\ast,i}Y_i}{\sum_{i=1}^n\hat\omega_i\hat W_{a,a^\ast,i}}. \label{eq:nuhat}
\end{equation}
The estimator of $\nu(1,0)$ thus uses the exposed individuals, reweighted by $1/\hat e_1$ to represent the whole population and by $\hat f(M\mid 0,\bX)/\hat f(M\mid 1,\bX)$ to give their mediator the distribution it would have under no exposure, and the estimator of $\nu(0,1)$ is constructed symmetrically from the unexposed \citep{huber2014identifying, lange2012simple}. The effect estimators are $\widehat\RR^{\mathrm{TE}}=\hat\nu(1,1)/\hat\nu(0,0)$, $\widehat\RR^{\mathrm{PDE}}=\hat\nu(1,0)/\hat\nu(0,0)$, $\widehat\RR^{\mathrm{PIE}}=\hat\nu(0,1)/\hat\nu(0,0)$ and $\widehat\RR^{\mathrm{InE}}=\hat\nu(1,1)\hat\nu(0,0)/\{\hat\nu(1,0)\hat\nu(0,1)\}$, so the decomposition holds exactly in every sample. The estimator (\ref{eq:nuhat}) is numerically identical to fitting the weighted log-linear marginal structural model $\log\E[Y\{a, M(a^\ast)\}]=\eta_0+\eta_1a$ by weighted estimating equations with weights $\hat\omega_i\hat W_{a,a^\ast,i}$ \citep{robins2000marginal}; the closed form makes the joint variance calculation below transparent. Algorithm~\ref{alg:3s} summarizes the procedure with the calibration estimator in Stage 1.

\begin{algorithm}[t]
\spacingset{1}
\caption{Three-stage weighting with moment calibration (3S-MC)}\label{alg:3s}
\begin{algorithmic}[1]
\Require Case-control data $(\bX_i,A_i,M_i,Y_i)$, $i=1,\ldots,n$; external moments $\hat{\bt}$ of univariate functions $\bh(\bX)$ with covariance $\mathbf V_t$
\State Compute $\hat\bmu_1,\hat\bmu_0,\hat\bSigma_1,\hat\bSigma_0$ and $\hat\bDelta$; set $\mathbf W$ diagonal
\For{$k=1,2,3$} \State $\hat\pi\gets\hat\bDelta^\top\mathbf W(\hat{\bt}-\hat\bmu_0)/\hat\bDelta^\top\mathbf W\hat\bDelta$;\quad $\mathbf W\gets\hat\bOmega(\hat\pi)^{-1}$ from (\ref{eq:Omega}) \EndFor
\State Report $\hat\pi$ and the overidentification statistic $J$; project $\hat\pi$ onto $[\epsilon,1-\epsilon]$ if necessary
\State Design weights $\hat\omega_i\gets\hat\pi/\lambda$ if $Y_i=1$ and $(1-\hat\pi)/(1-\lambda)$ if $Y_i=0$
\State Fit $e_1(\bX)$ and $P(M=1\mid A,\bX)$ by logistic regression with weights $\hat\omega_i$
\State Compute $\hat W_{a,a^\ast,i}$ from (\ref{eq:Wfun}) and $\hat\nu(a,a^\ast)$ from (\ref{eq:nuhat}) for $a,a^\ast\in\{0,1\}$
\State Form TE, PDE, PIE and InE; compute the stacked sandwich variance (\ref{eq:sandwich}) and intervals on the log scale
\end{algorithmic}
\end{algorithm}

\subsection{Asymptotic properties and variance estimation}\label{sec:asym}

We consider asymptotics in which $n\to\infty$ with $\lambda$ fixed and in which the external moments are either known or estimated independently of the case-control sample with $\sqrt n(\hat\bt-\bt)\to N(\mathbf 0,\bLambda)$. For a simple random external sample of size $N_{\mathrm{ext}}$ with $n/N_{\mathrm{ext}}\to\tau$, $\bLambda=\tau\Var\{\bh(\bX)\}$, and $\bLambda=\mathbf 0$ when $\bt$ is known. Let $\bSigma_y=\Var\{\bh(\bX)\mid Y=y\}$ and
\[
\bOmega=\frac{\pi^2\bSigma_1}{\lambda}+\frac{(1-\pi)^2\bSigma_0}{1-\lambda}+\bLambda .
\]

\begin{proposition}\label{prop:pi_asym}
Under Assumptions A1 to A3, if $\mathbf W$ is positive definite and $\hat{\mathbf W}\to\mathbf W$ in probability, then $\sqrt n(\hat\pi_{\hat{\mathbf W}}-\pi)\to N(0,\sigma^2_{\mathbf W})$ with
\[
\sigma^2_{\mathbf W}=\frac{\bDelta^\top\mathbf W\bOmega\mathbf W\bDelta}{(\bDelta^\top\mathbf W\bDelta)^2}\ \ge\ (\bDelta^\top\bOmega^{-1}\bDelta)^{-1},
\]
with equality when $\mathbf W=\bOmega^{-1}$, in which case $J\to\chi^2_{q-1}$ in distribution. If instead $\hat\bt\to\bt+\bm{\delta}$ for a fixed $\bm{\delta}\neq\mathbf 0$, then $\hat\pi_{\hat{\mathbf W}}\to\pi+\bDelta^\top\mathbf W\bm{\delta}/\bDelta^\top\mathbf W\bDelta$.
\end{proposition}

For a single calibration function and a design with equal numbers of cases and controls, the variance in Proposition~\ref{prop:pi_asym} is $2\{\pi^2\sigma_1^2+(1-\pi)^2\sigma_0^2\}/\Delta^2$ when $\bt$ is known, where $\sigma_y^2=\Var\{h(\bX)\mid Y=y\}$. The precision of $\hat\pi$ is therefore governed by the standardized separation $\Delta/\sigma_0$ between cases and controls, much as the precision of an instrumental variable estimator is governed by the strength of the instrument. The absolute standard error does not shrink as the outcome becomes rarer, so the relative precision of $\hat\pi$ deteriorates for rare outcomes, but by Proposition~\ref{prop:sens} the effect of a given relative error in $\pi$ on a log risk ratio is proportional to the counterfactual risk difference, which is small in exactly this situation. The last statement of Proposition~\ref{prop:pi_asym} quantifies the cost of misreported external moments: with one calibration function, an error $\delta$ in $t$ produces an error $\delta/\Delta$ in $\pi$, so external summaries must be accurate relative to the case-control separation. This is why we recommend moments from the same source population and period, several calibration functions, and inspection of the $J$ statistic.

For joint inference on all stages of 3S-MC, collect the parameters in
\[
\btheta=(\pi,\bm{\beta}^\top,\bm{\gamma}^\top,\bm{\nu}^\top)^\top,\qquad \bm{\nu}=\{\nu(0,0),\nu(1,0),\nu(0,1),\nu(1,1)\}^\top.
\]
The estimator solves $\sum_{i=1}^n\bpsi(O_i;\btheta,\hat\bt)=\mathbf 0$, where the estimating function stacks
\begin{align*}
\psi_\pi&=\omega(Y;\pi)\,\mathbf c^\top\{\bh(\bX)-\bt\},\\
\bpsi_\beta&=\omega(Y;\pi)\{A-e_1(\bX;\bm{\beta})\}\tilde\bX,\\
\bpsi_\gamma&=\omega(Y;\pi)\{M-P(M=1\mid A,\bX;\bm{\gamma})\}(\tilde\bX^\top,A)^\top,\\
\psi_{\nu(a,a^\ast)}&=\omega(Y;\pi)\,W_{a,a^\ast}(O;\bm{\beta},\bm{\gamma})\{Y-\nu(a,a^\ast)\},
\end{align*}
with $\mathbf c=\bOmega^{-1}\bDelta$ replaced by its estimate. Each component has mean zero under $\Pcc$ at the true parameter values, by Lemma~\ref{lem:rn} and Theorem~\ref{thm:eff}. Because cases and controls are sampled separately, the relevant central limit theorem is that for stratified samples. Define $\mathbf A=\Ecc\{\partial\bpsi/\partial\btheta^\top\}$, $\mathbf B=\lambda\Var\{\bpsi\mid Y=1\}+(1-\lambda)\Var\{\bpsi\mid Y=0\}$ and $\mathbf G=\Ecc\{\partial\bpsi/\partial\bt^\top\}$, whose only nonzero row is $-\mathbf c^\top$, the row of $\psi_\pi$.

\begin{theorem}\label{thm:asym}
Under Assumptions A1 to A3 and C1 to C3, correctly specified propensity score and mediator models, and standard regularity conditions for M-estimators, $\sqrt n(\hat\btheta-\btheta)\to N(\mathbf 0,\mathbf V)$ with
\begin{equation}
\mathbf V=\mathbf A^{-1}\left(\mathbf B+\mathbf G\bLambda\mathbf G^\top\right)\mathbf A^{-\top}. \label{eq:sandwich}
\end{equation}
For the log effects, $\sqrt n\{\log\widehat\RR-\log\RR\}\to N(\mathbf 0,\mathbf L\mathbf D\mathbf V_{\nu}\mathbf D\mathbf L^\top)$, where $\mathbf V_\nu$ is the block of $\mathbf V$ corresponding to $\bm{\nu}$, $\mathbf D=\mathrm{diag}(1/\bm{\nu})$ and $\mathbf L$ is the $4\times4$ contrast matrix whose rows map $\log\bm{\nu}$ to the logarithms of TE, PDE, PIE and InE.
\end{theorem}

The variance is estimated by replacing $\mathbf A$ by the derivative of the summed estimating functions at $\hat\btheta$, $\mathbf B$ by the sum over strata of the within-stratum outer products of the centred estimating functions, divided by $n$, and $\bLambda$ by $n\mathbf V_t$. Three features of (\ref{eq:sandwich}) deserve comment. First, the uncertainty in $\hat\pi$ automatically propagates into the effect estimates through the off-diagonal blocks of $\mathbf A^{-1}$. Second, the uncertainty of the external moments enters through $\mathbf G\bLambda\mathbf G^\top$ and is absent when $\bt$ is known. Third, using the unstratified variance $\Var_{\mathrm{cc}}(\bpsi)$ in place of $\mathbf B$ would add the nonnegative definite term $\lambda(1-\lambda)\{\Ecc(\bpsi\mid Y=1)-\Ecc(\bpsi\mid Y=0)\}^{\otimes2}$ and therefore give conservative intervals. A nonparametric alternative is the stratified bootstrap, in which cases and controls are resampled separately, $\hat\bt$ is redrawn from $N(\hat\bt,\mathbf V_t)$ when it is uncertain, and all three stages, including the estimation of $\pi$, are repeated in each bootstrap sample. For the density-ratio estimators, the bootstrap also resamples the reference or synthetic sample and repeats the classifier and risk-model fits, and it is the natural route to inference because a closed-form variance would require the influence functions of both working models.

\subsection{Frequency-matched designs}\label{sec:matched}

In a frequency-matched design, the controls are sampled within strata $C\in\{1,\ldots, K\}$ of matching variables, typically age group and sex, so that the distribution of $C$ among controls equals that among the cases. Assumption A1 is then replaced by the requirement that sampling is random within the outcome stratum for cases and within each combination of outcome and matching stratum for controls. Write $g_k=P(C=k\mid Y=1)$, which the cases estimate, $\lambda_1=n_1/n$, $\lambda_{0k}=n_{0k}/n$ for the sample fraction of controls in stratum $k$, and $\xi_k=P(C=k)$ for the population distribution of the matching variables, which is usually available from the same external source as the other covariate summaries.

\begin{proposition}\label{prop:matched}
In a frequency-matched design, $dP/d\Pcc(o)=\omega(y,k;\pi)$ with $\omega(1,k;\pi)=\pi/\lambda_1$ and $\omega(0,k;\pi)=(\xi_k-\pi g_k)/\lambda_{0k}$, provided $\xi_k>\pi g_k$ for all $k$. If $\bt=\E\{\bh(\bX)\}$ is known for calibration functions that are not functions of $C$ alone, then
\[
\bt-\sum_{k=1}^K\xi_k\bmu_{0k}=\pi\sum_{k=1}^Kg_k(\bmu_{1k}-\bmu_{0k}),\qquad\bmu_{yk}=\E\{\bh(\bX)\mid Y=y,C=k\},
\]
so that $\pi$ is identified whenever $\sum_kg_k(\bmu_{1k}-\bmu_{0k})\neq\mathbf 0$.
\end{proposition}

The matched analogue of 3S-MC replaces (\ref{eq:pihat}) by the corresponding generalized method of moments solution and the design weights by $\omega(Y_i, C_i;\hat\pi)$, while Stage 3 is unchanged except that the matching variables are included among the confounders when they are associated with exposure or mediator. The relevant separation is now the within-stratum difference between cases and controls, because matching removes the information carried by $C$ itself, and the matching variables should therefore not be used as calibration functions. Ignoring the matching and applying the unmatched design weights is not valid, since the controls no longer represent $P(\cdot\mid Y=0)$. The stacked estimating equations of Section~\ref{sec:asym} extend by adding the equations $I(Y=1)\{I(C=k)-g_k\}$ for $g_k$ and by defining the sampling strata as the cases and the $K$ strata of controls. Individually matched designs in which each case is matched to controls on discrete variables with a known population distribution fall within this framework once the matched sets are pooled into strata, an operation sometimes described as breaking the matching \citep{rose2009match, delcoigne2017breaking}; when matching is on continuous variables or on characteristics without an external distribution, the weights of Proposition~\ref{prop:matched} are not available and the conditional approach of \citet{vanderweele2016mediation}, which targets conditional odds ratio effects, remains the natural alternative.

\section{Simulation Study}\label{sec:sim}

\subsection{Design}\label{sec:simdesign}

The simulation study had four aims: to assess the bias, variability, and confidence interval coverage of the three estimators of the prevalence and of the four effects; to quantify the efficiency lost relative to an analysis in which the prevalence is known; to examine the robustness of the methods to the features discussed in Sections~\ref{sec:ident} and \ref{sec:est}; and to compare the sandwich and bootstrap variance estimators. We followed the structure recommended by \citet{morris2019using}. Table~\ref{tab:design} summarizes the data-generating mechanism of the base case and the variations considered.

\begin{table}[t]
\spacingset{1}\centering
\caption{Simulation design. Panel A gives the base case; Panel B lists the factors varied one at a time around the base case, each defining one scenario with $1000$ replications.}
\label{tab:design}
\small
\begin{tabular}{@{}p{3.3cm}p{12.4cm}@{}}
\toprule
\multicolumn{2}{@{}l}{\textit{A. Base case}}\\
\midrule
Covariates & $(Z_1,Z_2,Z_3)$ standard normal with exchangeable correlation $\rho=0.3$; $X_1=I(Z_1>z_{0.55})$, binary with mean $0.45$; $X_2=Z_2$, $X_3=Z_3$; all pairwise correlations positive \\
Exposure & $\mathrm{logit}\,P(A=1\mid\bX)=-0.5+\kappa(0.8X_1+0.6X_2-0.4X_3)$, $\kappa=1$ \\
Mediator & $\mathrm{logit}\,P\{M(a)=1\mid\bX\}=-0.4+0.9a+\kappa(0.7X_1+0.5X_2+0.3X_3)$ \\
Outcome & $\mathrm{logit}\,P\{Y(a,m)=1\mid\bX\}=\beta_0+0.9a+0.8m+0.5am+0.9X_1+0.7X_2+0.5X_3$, with $\beta_0$ set so that $\pi=0.10$ \\
Population and truth & Finite source population of $10^6$ individuals per scenario; true effects from the counterfactual risks averaged over the population \\
Sampling & $n_1=1000$ cases and $n_0=1000$ controls drawn at random from the population in each replication \\
External information & 3S-MC: population means of $X_1$, $X_2$, $X_3$, $X_2^2$, $X_3^2$ ($q=5$); 3S-DRJ: reference sample of $N=20{,}000$ drawn from the population; 3S-DRM: synthetic sample of $N_{\mathrm{syn}}=20{,}000$ from the product of the marginals (Bernoulli and normal) \\
\midrule
\multicolumn{2}{@{}l}{\textit{B. Factors varied (one at a time)}}\\
\midrule
Sample size & $n_1=n_0\in\{250,500,2000\}$; 1:3 control-to-case ratio with $n_1=500$, $n_0=1500$ \\
Prevalence & $\pi\in\{0.03,0.25\}$ \\
Correlation & $\rho\in\{0,0.6,0.9\}$ \\
Dimension & $p=8$ covariates, three of them binary (upper-tail dichotomization of the latent variables), $q=13$ moments \\
Confounding & $\kappa\in\{0.5,1.5\}$ \\
Misspecification & outcome model with added terms $0.6(X_2^2-1)+0.6X_2X_3$; exposure model with added terms $0.6(X_2^2-1)+0.5X_2X_3$; main-effects models fitted in both cases \\
External moments & estimated from a population reference sample of $N_{\mathrm{ext}}\in\{2000,500\}$; external mean of $X_2$ shifted by $d\in\{0.10,0.25\}$ standard deviations \\
Synthetic sample & $N_{\mathrm{syn}}\in\{2000,100{,}000\}$ \\
Matching & frequency matching on six strata formed by $X_1$ and the population tertiles of $X_2$, with $\xi_k$ known and moments of $X_2$, $X_3$ and their squares \\
\bottomrule
\end{tabular}
\end{table}

Five estimators were compared. The naive estimator applies Stage 3 to the unweighted case-control sample, and the oracle estimator uses the design weights with the true prevalence and therefore corresponds to the established known-prevalence weighting approach \citep{van2008estimation, rose2008simple}. The three estimators of Section~\ref{sec:stage1} differ only in Stage 1: 3S-MC uses the calibration moments of Table~\ref{tab:design} with the sandwich variance of Theorem~\ref{thm:asym}, and 3S-DRJ and 3S-DRM use main-effects logistic models for the risk function and the classifier, with the joint reference sample or the synthetic sample, respectively. 3S-DRJ was computed in the base case and in the scenarios with varying correlation, dimension, and outcome misspecification, and 3S-DRM in all unmatched scenarios. In the scenarios with estimated external moments, 3S-MC was also computed with the moments treated as fixed, and in the matched scenario the matched version of Proposition~\ref{prop:matched} was compared with the unmatched 3S-MC and with an oracle using the matched design weights. Performance was measured by the bias and empirical standard deviation of $\hat\pi$ and of the log effect estimators, the average estimated standard error, the coverage of nominal 95\% Wald intervals, and the relative efficiency with respect to the oracle, defined as the ratio of the empirical variances. With $1000$ replications, the Monte Carlo standard error of a coverage of 95\% is about $0.7$ percentage points. For the density-ratio estimators, we report point performance only, since their inference relies on the bootstrap.

\subsection{Results}\label{sec:simres}

Table~\ref{tab:prev} reports the prevalence results. Ignoring the design is not an option: the naive analysis treats the case fraction of $0.5$ as the prevalence and underestimates the total effect by a factor of about $0.55$ on the risk ratio scale, with coverage of at most $1.3\%$ in every scenario with a 1:1 design (Supplementary Table~S1). The calibration estimator 3S-MC was approximately unbiased in all scenarios in which the external moments were correct. Its largest bias, $0.010$ with $n_1=n_0=250$, was about a quarter of its standard deviation, and it decreased with the sample size at the rate expected for a ratio estimator. The average sandwich standard error agreed closely with the empirical standard deviation; in the base case they were $0.021$ and $0.020$, the latter being the value predicted by Proposition~\ref{prop:pi_asym}. Coverage ranged from $92.3\%$ to $95.9\%$ in the unmatched scenarios. Because positive correlation among the covariates reduces the variability of each calibration function given the others, the efficient combination of moments became more precise as the correlation increased, with the standard deviation of $\hat\pi$ falling from $0.024$ at $\rho=0$ to $0.017$ at $\rho=0.9$, as Proposition~\ref{prop:pi_asym} predicts. The overidentification test was close to its nominal level but slightly liberal, with rejection rates between $4.2\%$ and $7.3\%$ in the unmatched scenarios with correct external information. With $\pi=0.03$, the unconstrained estimate fell below zero in $5.6\%$ of the replications, reflecting the weak identification of a small prevalence discussed after Proposition~\ref{prop:pi_asym}; the constrained estimate was used in these cases, and the reported standard error and coverage for $\pi$ refer to the remaining $94.4\%$ of replications with an interior solution.

\begin{table}[t]
\spacingset{1}\centering
\caption{Estimation of the prevalence $\pi$ (1000 replications). Bias, empirical standard deviation (SD), and mean estimated standard error (SE) are on the probability scale; coverage (Cov) of nominal 95\% intervals, rejection rate of the overidentification test at the 5\% level ($J$), percentage of constrained solutions (Bnd), and percentage of replications in which the effects could not be computed (Fail) are in percent. SE and Cov for 3S-MC are computed over the replications with an interior solution. For 3S-DRM, bias and SD are computed over the replications in which (\ref{eq:fixedpoint}) had a solution and are not reported (--) when fewer than 50 replications had one.}
\label{tab:prev}
\footnotesize
\setlength{\tabcolsep}{3.8pt}
\begin{tabular}{lrrrrrrrrrrr}
\toprule
& & \multicolumn{7}{c}{3S-MC} & \multicolumn{3}{c}{3S-DRM} \\
\cmidrule(lr){3-9}\cmidrule(lr){10-12}
Scenario & $\pi$ & Bias & SD & SE & Cov & $J$ & Bnd & Fail & Bias & SD & Fail \\
\midrule
Base case & 0.100 & 0.001 & 0.021 & 0.020 & 95.6 & 5.5 & 0.0 & 0.0 & -0.081 & 0.013 & 25.7 \\
$n_1=n_0=250$ & 0.100 & 0.010 & 0.039 & 0.040 & 92.3 & 7.3 & 0.1 & 0.0 & -0.072 & 0.020 & 44.0 \\
$n_1=n_0=500$ & 0.100 & 0.005 & 0.029 & 0.029 & 93.2 & 6.7 & 0.1 & 0.0 & -0.078 & 0.016 & 42.8 \\
$n_1=n_0=2000$ & 0.100 & 0.001 & 0.015 & 0.014 & 94.0 & 5.2 & 0.0 & 0.0 & -0.074 & 0.012 & 3.1 \\
$\pi=0.03$ & 0.030 & 0.004 & 0.019 & 0.020 & 95.0 & 4.4 & 5.6 & 0.0 & -- & -- & 99.7 \\
$\pi=0.25$ & 0.250 & 0.002 & 0.020 & 0.019 & 94.2 & 5.2 & 0.0 & 0.0 & -0.081 & 0.025 & 0.0 \\
$\rho=0$ & 0.100 & 0.001 & 0.024 & 0.025 & 95.8 & 4.6 & 0.0 & 0.0 & 0.006 & 0.021 & 0.0 \\
$\rho=0.6$ & 0.100 & 0.001 & 0.018 & 0.018 & 93.9 & 5.3 & 0.0 & 0.0 & -- & -- & 99.5 \\
$\rho=0.9$ & 0.099 & 0.002 & 0.017 & 0.017 & 94.8 & 6.0 & 0.0 & 0.0 & -- & -- & 100.0 \\
$p=8$ covariates & 0.100 & 0.005 & 0.017 & 0.017 & 93.8 & 5.6 & 0.0 & 0.0 & -- & -- & 99.9 \\
Weak confounding & 0.101 & 0.003 & 0.022 & 0.022 & 94.3 & 4.9 & 0.0 & 0.0 & -0.082 & 0.013 & 40.3 \\
Strong confounding & 0.100 & 0.002 & 0.020 & 0.020 & 94.4 & 4.9 & 0.0 & 0.0 & -0.078 & 0.013 & 15.4 \\
Misspecified outcome model & 0.101 & 0.002 & 0.010 & 0.010 & 95.9 & 4.5 & 0.0 & 0.0 & -- & -- & 97.9 \\
Misspecified exposure model & 0.100 & 0.002 & 0.018 & 0.018 & 95.0 & 5.6 & 0.0 & 0.0 & -0.083 & 0.011 & 21.2 \\
$\hat{\mathbf t}$ from survey, $N_{\mathrm{ext}}=2000$ & 0.101 & 0.002 & 0.027 & 0.027 & 94.9 & 6.1 & 0.1 & 0.0 & -0.081 & 0.014 & 24.3 \\
$\hat{\mathbf t}$ from survey, $N_{\mathrm{ext}}=500$ & 0.101 & 0.002 & 0.042 & 0.041 & 95.2 & 4.2 & 1.0 & 0.0 & -0.081 & 0.013 & 23.1 \\
Biased mean, $d=0.10$ & 0.100 & 0.051 & 0.020 & 0.020 & 27.9 & 73.5 & 0.0 & 0.0 & -0.080 & 0.013 & 22.4 \\
Biased mean, $d=0.25$ & 0.100 & 0.123 & 0.020 & 0.019 & 0.0 & 100.0 & 0.0 & 0.0 & -0.080 & 0.013 & 21.6 \\
Ratio 1:3 ($n_1=500$) & 0.100 & 0.001 & 0.017 & 0.017 & 94.8 & 4.5 & 0.0 & 0.0 & -0.025 & 0.019 & 0.0 \\
$N_{\mathrm{syn}}=2000$ & 0.101 & 0.002 & 0.021 & 0.020 & 94.0 & 5.6 & 0.0 & 0.0 & -0.023 & 0.023 & 0.1 \\
$N_{\mathrm{syn}}=100000$ & 0.100 & 0.001 & 0.020 & 0.020 & 93.8 & 4.6 & 0.0 & 0.0 & -0.088 & 0.009 & 55.3 \\
Frequency matched & 0.100 & 0.007 & 0.036 & 0.033 & 92.1 & 12.2 & 0.0 & 0.5 & -- & -- & -- \\
\bottomrule
\end{tabular}
\end{table}

The density-ratio estimators performed very differently. With independent covariates, 3S-DRM had a small bias ($0.006$, Table~\ref{tab:prev}), but as soon as the covariates were correlated, it broke down: with $\rho=0.3$ its bias was $-0.081$ and (\ref{eq:fixedpoint}) had no solution in $26\%$ of the replications, and with $\rho=0.6$ and $\rho=0.9$ a solution existed in fewer than $1\%$ of the replications. Proposition~\ref{prop:fixed} explains the source of the problem. Under the design of Table~\ref{tab:design}, the derivative $\ell_Q'(0+)$ of Proposition~\ref{prop:fixed}(i), computed for the product of the true marginals, equals $0.186$ at $\rho=0$, where $\pi_Q=0.100$, but $-0.015$, $-0.179$ and $-0.313$ at $\rho=0.3$, $0.6$ and $0.9$, so the target $\pi_Q$ lies on the boundary and, with correctly specified working models, the fixed-point equation would have no interior solution (Supplementary Table~S5). The implemented estimator uses main-effects working models, and with a misspecified classifier its large-sample limit depends on the ratio $N_{\mathrm{syn}}/n$ of the synthetic to the case-control sample size. At the base-case ratio $10$ this limit is $0.104$ at $\rho=0$ and $0.017$ at $\rho=0.3$ and does not exist at $\rho=0.6$ and $0.9$, and at $\rho=0.3$ it decreases from $0.083$ at $N_{\mathrm{syn}}/n=1$ to $0.003$ at $N_{\mathrm{syn}}/n=50$ (Supplementary Table~S5). This single mechanism accounts for the failure pattern in Table~\ref{tab:prev}. Since $N_{\mathrm{syn}}=20{,}000$ in all but two scenarios, the ratio increased as $n$ decreased, and the proportion of replications without a solution rose from $0.1\%$ at ratio $1$ ($N_{\mathrm{syn}}=2000$) and $3\%$ at ratio $5$ ($n_1=n_0=2000$) to $55\%$ at ratio $50$ ($N_{\mathrm{syn}}=100{,}000$), because the closer the limit lies to zero, the more often sampling variability pushes the root outside the search interval $(0.001,0.999)$. Among the replications with a solution, the biases at the ratios $1$ and $10$ ($-0.023$ and $-0.081$) are close to the differences between the limits and $\pi$ ($-0.017$ and $-0.083$). 3S-DRJ, whose reference distribution is the population distribution, was close to unbiased at every correlation, with a small positive bias between $0.006$ and $0.013$ that is of the same order as the excess over $\pi$ of its large-sample limits at the ratio $10$ used in the simulation ($0.103$ to $0.113$, Supplementary Table~S5) and reflects the main-effects risk model, since the marginal risk function obtained after integrating over the exposure and mediator is not exactly logistic. With nonlinear terms in the outcome model, this misspecification became important, and the bias of 3S-DRJ reached $-0.057$. Figure~\ref{fig:prev} displays these contrasts. 3S-MC, which uses neither a risk model nor a classifier, was unaffected by covariate dependence and by outcome misspecification.

\begin{figure}[t]
\centering
\includegraphics[width=\textwidth]{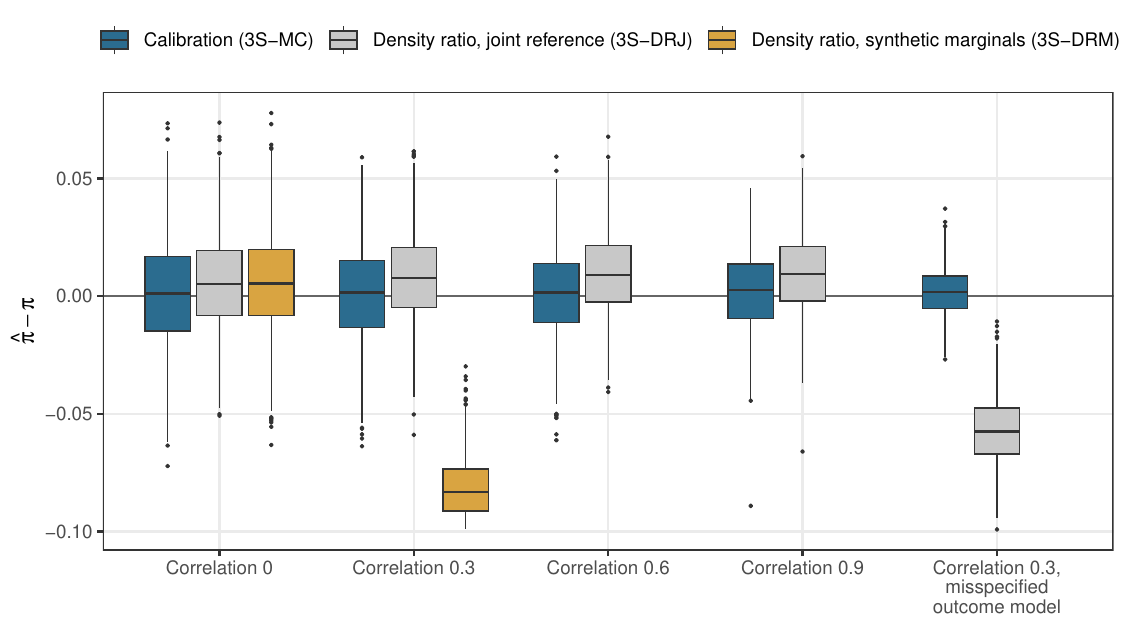}
\caption{Error $\hat\pi-\pi$ of the three estimators of the prevalence across covariate correlations and under a misspecified outcome model (1000 replications each, true $\pi=0.10$). Only replications in which an estimator had a solution are shown, and boxes are omitted when fewer than 50 replications had one. 3S-DRM had no solution in $26\%$, $99.5\%$ and $100\%$ of the replications with correlation $0.3$, $0.6$ and $0.9$, and in $98\%$ with the misspecified outcome model.}
\label{fig:prev}
\end{figure}

Table~\ref{tab:eff} reports the results for the effects estimated with 3S-MC in Stage 1. With correct external information and a correctly specified exposure model, the bias of the log effect estimators was at most $0.027$ in absolute value, small relative to their standard deviations, and coverage was between $92.9\%$ and $98.1\%$ for all four effects. The highest coverages occurred for the mediated interaction, both for 3S-MC and for the oracle (Supplementary Table~S2). Two features of the sandwich standard error contribute to this. First, its average exceeded the empirical standard deviation of the log mediated interaction in some scenarios: the ratio of the two ranged from $0.953$ to $1.077$ and was $1.071$ in the base case, where it alone would reduce the standard deviation of the Wald statistic to $0.933$. Second, the estimated standard error was positively correlated with the absolute estimation error (correlation $0.26$ in the base case), which reduced that standard deviation further to the observed $0.910$; this second feature explains why the intervals over-covered also in scenarios in which the ratio was below one, such as $n_1=n_0=250$, with ratio $0.988$ and coverage $98.1\%$. Apart from this, the results are very close to those of the oracle (Supplementary Tables~S1 and S2), so the propagation of the uncertainty about $\pi$ through the stacked estimating equations is accurate. The relative efficiency of the total effect estimator with respect to the oracle varied between $0.68$ and $0.96$ when the external moments were known exactly and correct (Supplementary Table~S3 gives it for all four effects). As predicted by Proposition~\ref{prop:sens}, it was higher for the rarer outcome ($0.91$ for $\pi=0.03$) than for the common outcome ($0.82$ for $\pi=0.25$); it was highest when the outcome model was misspecified ($0.96$), because the covariates then discriminate strongly between cases and controls and $\pi$ is estimated precisely, and it was lowest for the frequency-matched design ($0.68$), because matching removes part of the discriminating information. When the exposure model was misspecified, all weighting estimators, including the oracle, were biased, which is a property of Stage 3 rather than of the prevalence calibration.

\begin{table}[t]
\spacingset{1}\centering
\caption{Performance of the effect estimators with moment calibration (3S-MC) in Stage 1 (1000 replications). Bias is that of the log risk ratio multiplied by 100, Cov is the coverage of nominal 95\% intervals in percent, and RE is the ratio of the empirical variance of the oracle estimator, which uses the true prevalence, to that of 3S-MC.}
\label{tab:eff}
\footnotesize
\begin{tabular}{lrrrrrrrrrr}
\toprule
& \multicolumn{4}{c}{Total effect} & \multicolumn{2}{c}{PDE} & \multicolumn{2}{c}{PIE} & \multicolumn{2}{c}{InE} \\
\cmidrule(lr){2-5}\cmidrule(lr){6-7}\cmidrule(lr){8-9}\cmidrule(lr){10-11}
Scenario & $\RR$ & Bias & Cov & RE & Bias & Cov & Bias & Cov & Bias & Cov \\
\midrule
Base case & 3.16 & -0.4 & 95.0 & 0.87 & -0.5 & 95.4 & -0.0 & 94.7 & 0.1 & 96.9 \\
$n_1=n_0=250$ & 3.16 & -0.6 & 95.5 & 0.91 & 0.0 & 95.6 & 0.2 & 93.2 & -0.8 & 98.1 \\
$n_1=n_0=500$ & 3.16 & -2.2 & 95.2 & 0.87 & -1.8 & 95.9 & 0.1 & 92.9 & -0.5 & 96.4 \\
$n_1=n_0=2000$ & 3.16 & -1.5 & 94.6 & 0.84 & -1.6 & 95.0 & -0.0 & 94.9 & 0.2 & 95.3 \\
$\pi=0.03$ & 3.84 & -2.5 & 96.1 & 0.91 & -2.7 & 95.9 & 0.1 & 95.6 & 0.1 & 97.8 \\
$\pi=0.25$ & 2.34 & -0.3 & 95.2 & 0.82 & -0.3 & 94.8 & -0.2 & 93.9 & 0.1 & 96.2 \\
$\rho=0$ & 3.32 & 1.5 & 95.4 & 0.81 & 1.3 & 95.3 & -0.0 & 94.0 & 0.2 & 96.7 \\
$\rho=0.6$ & 3.05 & -0.9 & 94.1 & 0.88 & -1.0 & 93.7 & 0.0 & 94.5 & 0.1 & 96.9 \\
$\rho=0.9$ & 3.00 & 0.1 & 96.6 & 0.88 & 0.4 & 96.5 & -0.1 & 94.2 & -0.2 & 95.8 \\
$p=8$ covariates & 3.00 & -1.8 & 94.5 & 0.95 & -1.9 & 94.8 & 0.4 & 95.3 & -0.3 & 96.0 \\
Weak confounding & 3.12 & -0.1 & 95.2 & 0.78 & 0.1 & 95.2 & -0.4 & 94.2 & 0.3 & 96.0 \\
Strong confounding & 3.18 & -0.8 & 94.9 & 0.88 & -0.8 & 95.2 & 0.1 & 93.7 & -0.1 & 96.3 \\
Misspecified outcome model & 2.10 & -1.1 & 94.4 & 0.96 & -1.0 & 94.9 & -0.2 & 94.2 & 0.1 & 96.5 \\
Misspecified exposure model & 3.18 & 38.7 & 6.2 & 0.86 & 40.3 & 5.6 & 2.8 & 90.7 & -4.3 & 70.9 \\
$\hat{\mathbf t}$ from survey, $N_{\mathrm{ext}}=2000$ & 3.15 & -1.2 & 96.3 & 0.76 & -1.1 & 96.6 & 0.1 & 95.3 & -0.2 & 96.0 \\
$\hat{\mathbf t}$ from survey, $N_{\mathrm{ext}}=500$ & 3.15 & 0.9 & 95.4 & 0.61 & 1.1 & 95.3 & 0.1 & 95.9 & -0.3 & 96.7 \\
Biased mean, $d=0.10$ & 3.16 & -9.9 & 85.7 & 1.02 & -9.3 & 87.2 & -0.8 & 90.9 & 0.1 & 96.2 \\
Biased mean, $d=0.25$ & 3.16 & -22.1 & 41.3 & 1.20 & -21.0 & 46.2 & -0.9 & 91.7 & -0.2 & 95.7 \\
Ratio 1:3 ($n_1=500$) & 3.15 & -1.8 & 95.2 & 0.91 & -1.7 & 94.8 & 0.1 & 96.1 & -0.3 & 95.8 \\
$N_{\mathrm{syn}}=2000$ & 3.16 & -0.3 & 95.3 & 0.84 & 0.1 & 95.1 & -0.0 & 94.6 & -0.4 & 96.4 \\
$N_{\mathrm{syn}}=100000$ & 3.16 & 0.2 & 94.5 & 0.80 & 0.4 & 94.6 & -0.0 & 93.9 & -0.3 & 96.2 \\
Frequency matched & 3.16 & -1.1 & 93.2 & 0.68 & -0.4 & 93.8 & 0.6 & 93.1 & -1.3 & 97.9 \\
\bottomrule
\end{tabular}
\end{table}

The treatment of external uncertainty mattered. When the external moments were estimated from a reference sample of $2000$ or $500$ individuals, including their covariance in the sandwich gave coverage of $94.9\%$ and $95.2\%$ for $\pi$ and $96.3\%$ and $95.4\%$ for the total effect, whereas treating the estimated moments as fixed gave coverage of $86.1\%$ and $64.1\%$ for $\pi$ and $95.2\%$ and $90.9\%$ for the total effect. The scenarios with biased external moments illustrate the limits of the method. A shift of $d=0.10$ standard deviations in a single external mean produced a bias of $0.051$ in $\hat\pi$, close to the value $0.048$ given by Proposition~\ref{prop:pi_asym} with the efficient weight matrix, and reduced the coverage for the total effect to $86\%$; the overidentification test rejected in $74\%$ of the replications. With $d=0.25$, the bias was $0.123$, the coverage for the total effect was $41\%$, and the test rejected in every replication.

In the frequency-matched scenario, applying the unmatched design weights produced prevalence estimates on the boundary of the parameter space in $81\%$ of the replications and a coverage of $40\%$ for the total effect. The matched version of Proposition~\ref{prop:matched} was approximately unbiased (bias $0.007$), with coverage of $92.1\%$ for $\pi$ and $93.2\%$ for the total effect, and the effects could not be computed in $0.5\%$ of the replications, in which the estimated prevalence made the control weight $\xi_k-\hat\pi\hat g_k$ of a matching stratum negative. The overidentification test was rejected in $12\%$ of the replications, and the coverage for $\pi$ was somewhat below nominal. Both shortfalls probably arise because the numbers of controls per stratum are fixed in proportion to the observed case distribution, whereas the sandwich variance conditions on them and therefore ignores their dependence on $\hat g_k$; in matched designs the $J$ statistic should accordingly be interpreted with caution.

Finally, in a separate check under the base case with $200$ replications, each analysed with $200$ stratified bootstrap samples, the average sandwich and bootstrap standard errors agreed to within about $3\%$ for the prevalence and for every log effect (for example $0.123$ and $0.124$ for the total effect and $0.031$ and $0.031$ for the pure indirect effect), and the coverage of the corresponding intervals was similar (Supplementary Table~S4). Either method can therefore be used with 3S-MC; the sandwich is much cheaper.

\section{Application to NHANES Data}\label{sec:app}

To illustrate the methods in a setting where the answer is known, we embedded a case-control study in a well-characterized population. We used the 2017--2018 cycle of the National Health and Nutrition Examination Survey \citep{nhanes_cdc}, linking the demographic, smoking, blood pressure, medical conditions, and body measures files. The outcome was a self-reported physician diagnosis of coronary heart disease, myocardial infarction, stroke, or congestive heart failure; the exposure was current cigarette smoking; the mediator was a self-reported diagnosis of hypertension; and the confounders were age, sex, race and ethnicity, and body mass index. The $N=5140$ adults with complete data were treated as the target population; survey weights were not used, so the estimands refer to this analytic cohort rather than to the population of the United States. In the cohort, the prevalence of the outcome was $0.111$; $18.2\%$ were current smokers and $37.9\%$ reported hypertension. The data are cross-sectional, and the variables are self-reported, so the temporal ordering of smoking, hypertension and cardiovascular disease cannot be guaranteed, and current smoking may partly reflect cessation after a cardiovascular event. The analysis should therefore be read as a methodological illustration rather than as a substantive assessment of the mechanism linking smoking and cardiovascular disease.

The benefit of this design is that the effects can be estimated in the full cohort, where no prevalence is needed, which provides a benchmark for the case-control analysis. In the full cohort, inverse probability weighting with the same logistic propensity score and mediator models gave $\RR^{\mathrm{TE}}=1.23$ (95\% confidence interval $1.01$ to $1.50$), $\RR^{\mathrm{PDE}}=1.19$ ($0.98$ to $1.45$), $\RR^{\mathrm{PIE}}=1.03$ ($1.01$ to $1.06$) and $\RR^{\mathrm{InE}}=1.00$ ($0.98$ to $1.02$). We then formed a case-control sample consisting of all $573$ cases and $573$ randomly selected controls. For 3S-MC, the external information consisted of the cohort means of age, age squared, body mass index and its square, and the proportions of women and of the four non-reference race and ethnicity groups, giving $q=9$ univariate moments, whose covariance was taken to be that of a sample mean of $N$ observations. For 3S-DRJ, a reference sample of $30{,}000$ covariate vectors was drawn with replacement from the cohort, and for 3S-DRM each covariate was drawn independently from its marginal distribution in the cohort.

The calibrated prevalence was $\hat\pi=0.123$ with a sandwich 95\% interval of $0.059$ to $0.187$, and the stratified bootstrap with redrawn external moments gave a standard error of $0.033$ and a percentile interval of $0.062$ to $0.191$, both covering the cohort prevalence of $0.111$. The overidentification statistic was $J=8.24$ on $8$ degrees of freedom ($p=0.41$), giving no indication that the external moments were incompatible with the sample. The wide interval reflects the moderate discrimination between cases and controls achieved by the available covariates, of which age is by far the most informative. The density-ratio estimators gave $0.136$ with the joint reference sample and $0.114$ with the synthetic marginal sample. The covariates in this cohort are only weakly correlated; the largest absolute correlation between distinct covariates is $0.17$ (between body mass index and one race and ethnicity indicator). Both density-ratio estimates lie within the interval of 3S-MC, and their difference from each other and from the cohort value is within sampling error, so that a single sample provides no evidence for or against Proposition~\ref{prop:fixed}(iii).

Table~\ref{tab:nhanes} compares the case-control estimates with the benchmark. With 3S-MC in Stage 1, the total effect was estimated as $1.16$ ($0.87$ to $1.54$), the pure direct effect as $1.13$ ($0.85$ to $1.50$), the pure indirect effect through hypertension as $1.03$ ($0.96$ to $1.10$) and the mediated interaction as $1.00$ ($0.98$ to $1.02$). All intervals contain the benchmark values, and the estimates are almost identical to those of the oracle analysis that uses the cohort prevalence. This is consistent with Proposition~\ref{prop:sens}: the estimated counterfactual risks were $0.119$ and $0.138$ without and with exposure, so the risk difference is small and the risk ratios hardly depend on the prevalence. The naive analysis, which ignores the design, attenuated the total effect to $1.09$. The sandwich and bootstrap standard errors of the log effects agreed closely for the total, direct and indirect effects ($0.146$ and $0.152$ for the total effect, $0.033$ and $0.034$ for the pure indirect effect); for the mediated interaction the bootstrap distribution had a longer lower tail and a standard error of $0.024$ compared with $0.012$.

\begin{table}[t]
\spacingset{1}\centering
\caption{NHANES 2017--2018 application: risk ratio estimates (95\% confidence intervals) for the effect of current smoking on cardiovascular disease, with hypertension as mediator. The benchmark uses the full cohort of $N=5140$; the other analyses use a 1:1 case-control sample of $1146$ individuals. Intervals for 3S-MC are based on the sandwich variance; the oracle uses the cohort prevalence $0.111$.}
\label{tab:nhanes}
\small
\begin{tabular}{lcccc}
\toprule
Effect & Benchmark (cohort) & 3S-MC & Oracle ($\pi$ known) & Naive \\
\midrule
TE & 1.23 (1.01, 1.50) & 1.16 (0.87, 1.54) & 1.16 (0.87, 1.55) & 1.09 (0.95, 1.26) \\
PDE & 1.19 (0.98, 1.45) & 1.13 (0.85, 1.50) & 1.13 (0.84, 1.51) & 1.08 (0.94, 1.23) \\
PIE & 1.03 (1.01, 1.06) & 1.03 (0.96, 1.10) & 1.03 (0.96, 1.10) & 1.02 (0.99, 1.05) \\
InE & 1.00 (0.98, 1.02) & 1.00 (0.98, 1.02) & 1.00 (0.98, 1.02) & 1.00 (0.99, 1.01) \\
\midrule
Prevalence & 0.111 & 0.123 (0.059, 0.187) & 0.111 & 0.500 \\
\bottomrule
\end{tabular}
\end{table}

Figure~\ref{fig:nhanes} shows how the estimates change when the analysis is repeated with the prevalence fixed at values between $0.04$ and $0.25$. Over this wide range, the total effect moves only from $1.18$ to $1.13$ and the pure indirect effect from $1.031$ to $1.024$. The curves obtained by refitting all models at each value coincided to the third decimal with the closed-form prediction of Proposition~\ref{prop:sens}(ii) computed from the counterfactual risks at $\hat\pi$, so that in practice a sensitivity analysis over the prevalence requires no refitting.

\begin{figure}[t]
\centering
\includegraphics[width=0.8\textwidth]{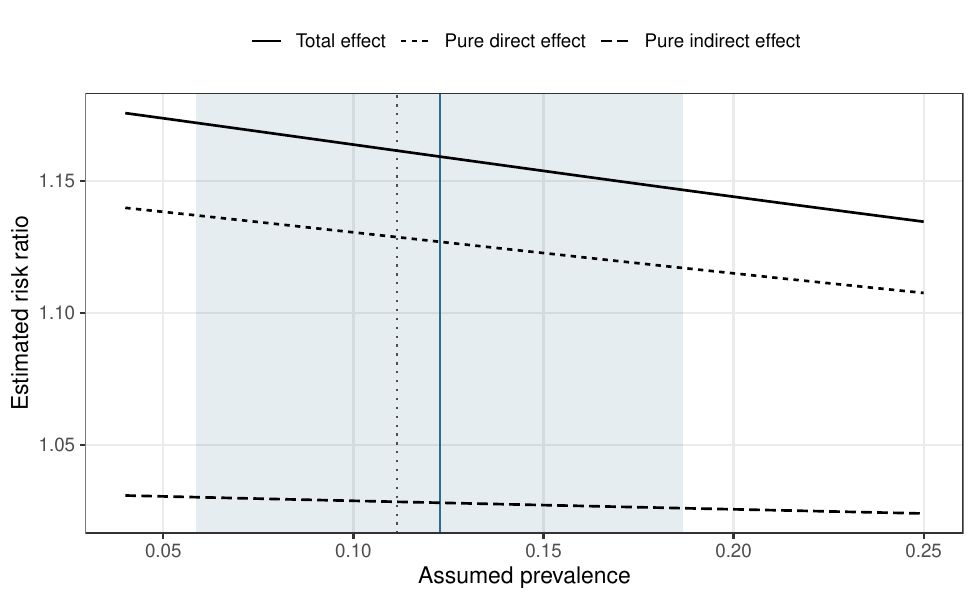}
\caption{NHANES application: estimated total, pure direct, and pure indirect effects when the design weights are computed with an assumed prevalence. The solid vertical line and shaded band show the calibrated prevalence and its 95\% confidence interval; the dotted line shows the prevalence in the cohort.}
\label{fig:nhanes}
\end{figure}

Finally, we repeated the case-control sampling $500$ times, drawing $400$ of the $573$ cases together with either $400$ or $1600$ controls. With $1600$ controls, the average prevalence estimate was $0.113$, and the geometric mean of the total effect estimates was $1.25$, against benchmark values of $0.111$ and $1.23$. With $400$ controls, the prevalence estimates averaged $0.122$ and the intervals for the prevalence covered the cohort value in $94.6\%$ of the samples, but the total effect estimates had a geometric mean of $1.33$. The oracle analysis showed the same upward shift ($1.33$), so this is a finite-sample property of design-weighted inverse probability weighting when each control represents many members of the population, not a consequence of estimating the prevalence, and it supports the familiar advice to sample several controls per case when absolute or marginal effects are of interest. Because the repeated samples are drawn without replacement from a finite cohort and share most of the cases, their variability understates that of independent studies, and the coverage of the benchmark effects by the confidence intervals, between $96\%$ and $100\%$, is correspondingly conservative.

\section{Discussion}\label{sec:disc}

Case-control data have long been regarded as uninformative about absolute and marginal quantities unless the outcome prevalence is supplied from outside the study. This article has shown that the information required need not be the prevalence itself. Because case-control sampling acts on the source population only through the outcome, the reweighting that restores the population is known up to a single scalar, the mixing weight of the case and control distributions, and external information on the covariates determines that scalar. Marginal moments identify it through a linear equation, which leads to the calibration estimator 3S-MC, valid whatever the dependence among the covariates and requiring neither an outcome model nor a reference sample. A reference distribution identifies it as the maximizer of a mixture likelihood, which leads to the density-ratio estimators; Proposition~\ref{prop:fixed} shows that they target the prevalence when the reference distribution is the population distribution and an independence approximation of it otherwise. Once the prevalence is estimated, the design weights turn the case-control sample into a pseudo-population in which the total effect and its decomposition into pure direct, pure indirect, and mediated interaction components are estimated by inverse probability weighting, and for 3S-MC the stacked estimating equations deliver a single sandwich variance that respects the stratified design and propagates the uncertainty of both the prevalence and the external moments.

The approach complements the likelihood-based mediation methods of \citet{satten2022efficient} and \citet{doretti2024mediation}. Those methods exploit the parametric structure of the outcome and mediator models, which yields efficiency gains, and they target natural effects on the odds ratio scale, either under a rare-disease approximation or with a known prevalence. The present framework targets risk ratio effects, makes no rare-disease assumption, and replaces knowledge of the prevalence by knowledge of covariate summaries; the price is the efficiency loss associated with estimating the prevalence, which in our simulations ranged from 4 to 32 percent of the variance of the total effect when the external moments were known and correct and was larger when they were themselves estimated from a small survey, and the reliance on inverse probability weighting. Proposition~\ref{prop:sens} also clarifies an old intuition: when the outcome is rare, ratio effects depend only weakly on the prevalence, because the effect of a relative error in $\pi$ is damped by the small counterfactual risk difference, so the main role of the calibration step in that case is a principled propagation of the remaining uncertainty, whereas for common outcomes the prevalence matters and an incorrect value can bias the risk ratios substantially. When an investigator does know the prevalence, it can be used directly, and when it is known only approximately, it can be added as one more moment condition with its own uncertainty, in which case the overidentification test assesses its compatibility with the covariate summaries.

Several limitations deserve emphasis. The methods stand or fall with the transportability of the external information, which must describe the population that generated the cases and controls. Proposition~\ref{prop:pi_asym} shows that an error in an external mean translates into an error in the prevalence that is inversely proportional to the separation between cases and controls, and our simulations confirm that modest discrepancies can matter. The overidentification test provides a partial safeguard but is slightly liberal in finite samples and cannot detect errors aligned with the case-control contrast, and we therefore recommend summaries from the same population and period, several calibration functions, and a sensitivity analysis over the external moments or, by Proposition~\ref{prop:sens}, directly over the prevalence. The density-ratio estimators depend on two working models, and the version based on a synthetic marginal sample additionally on approximate independence of the covariates; we regard them as complements to 3S-MC rather than as replacements, and their inference in this paper relies on the bootstrap. The precision of the calibrated prevalence depends on how strongly the available covariates discriminate between cases and controls, and with weak discrimination or very rare outcomes the unconstrained estimate may fall outside the unit interval. As with all weighting estimators, Stage 3 is sensitive to misspecification of the propensity score and mediator models and to extreme weights. Doubly or multiply robust estimators of natural effects \citep{tchetgen2012semiparametric} can be combined with the design weights of Stage 2 without modification, and developing the corresponding joint inference is a natural next step. Finally, the identification argument assumes that controls are sampled from the source population of the cases; it does not address selection bias arising from differential participation, which requires separate modelling.

The framework extends in several directions. Calibration functions that depend on several covariates, such as the proportion of women above a given age, can be added whenever their population values are known, and more informative external information, including covariate-specific prevalences \citep{qin2015using}, can be accommodated by additional moment conditions. A further extension concerns nested case-control studies with time-varying exposures, in which controls are sampled from the risk set at each event time. The design weights then become functions of the risk set and of time, and external information on the age- and period-specific structure of the population could play the role of the marginal moments used here. Because the population quantity of interest is then a hazard or a cumulative incidence rather than a single prevalence, this extension requires its own identification analysis, and we leave it to future work.

\section*{Supplementary Materials}
\sloppy

\noindent\textbf{Supplement:} PDF file containing the proofs of all results (Section S1), complete simulation results for every estimator and scenario (Section S2, Tables S1 to S3), the comparison of sandwich and bootstrap variance estimators (Section S3, Table S4), the large-sample limits of the density-ratio estimators (Section S4, Table S5), and a description of the code and data (Section S5).

\noindent\textbf{Code and data:} Archive containing R code that reproduces all simulation results, tables, and figures (parent file \texttt{00\_RUN\_SIMULATION.R}) and the application (parent file \texttt{00\_RUN\_APPLICATION.R}), together with the analytic NHANES data set, a script documenting its construction from the public files, and a README file.

\section*{Data Availability Statement}

The NHANES 2017--2018 data are publicly available from the National Center for Health Statistics \citep{nhanes_cdc}. The analytic data set and all code are provided in the Supplementary Materials. Computations were carried out in R \citep{R2024} using the \texttt{ggplot2} package \citep{wickham2016ggplot2}; the data were retrieved with the \texttt{nhanesA} package \citep{ale2024nhanesa}.

\section*{Disclosure Statement}

The authors report no competing interests to declare.

\bibliographystyle{agsm}
%\bibliography{bibfile}

\end{document}